\documentclass[onecolumn, draftcls,12pt]{IEEEtran}

\IEEEoverridecommandlockouts                              	
\usepackage{amsmath,amsfonts,amssymb,amscd}
\usepackage{algorithm,algorithmic}
\usepackage{graphics} 
\usepackage{graphicx}
\usepackage{epsfig} 
\usepackage{float}
\usepackage{cite}
\usepackage{verbatim}
\usepackage{url,changebar,bm,xspace,dsfont}
\usepackage[usenames,dvipsnames]{color}

\newtheorem{remark}{\bfseries Remark}

\newtheorem{defn}{\bfseries Definition}
\newtheorem{prop}{\bfseries Proposition}

\let\mathbb=\mathds 
\def\diag{\mathop{\mathrm{diag}}}  

\floatstyle{ruled}
\newfloat{model}{H}{mod}
\floatname{model}{\footnotesize Model}
\newfloat{notatio}{H}{not}
\floatname{notatio}{\footnotesize Notation}

\begin{document}

%
%
%
%
\title{\LARGE \bf Distributed Formation of Balanced and Bistochastic \\ Weighted Diagraphs in Multi-Agent Systems
}

\author{Themistoklis~Charalambous\thanks{Themistoklis~Charalambous is with the School of Electrical Engineering, Royal Institute of Technology (KTH), Stockholm, Sweden. E-mail: {\tt themisc@kth.se}.} and Christoforos~N.~Hadjicostis
        \thanks{Christoforos~N.~Hadjicostis is with the Department of Electrical and Computer Engineering at the University of Cyprus, Nicosia, Cyprus. E-mail:{\tt~ chadjic@ucy.ac.cy}.}
}
\maketitle
\thispagestyle{empty}
\pagestyle{empty}

%
%
%
%
\begin{abstract}
Consensus strategies find a variety of applications in distributed coordination and decision making in multi-agent systems. In particular, average consensus plays a key role in a number of applications and is closely associated with two classes of digraphs, weight-balanced (for continuous-time systems) and bistochastic (for discrete-time systems). A weighted digraph is called balanced if, for each node $v_j$, the sum of the weights of the edges outgoing from that node is equal to the sum of the weights of the edges incoming to that node. In addition, a weight-balanced digraph is bistochastic if all weights are nonnegative and, for each node $v_j$, the sum of weights of edges incoming to that node and the sum of the weights of edges out-going from that node is unity; this implies that the corresponding weight matrix is column and row stochastic (i.e., doubly stochastic). We propose two distributed algorithms: one solves the weight-balance problem and the other solves the bistochastic matrix formation problem for a distributed system whose components (nodes) can exchange information via interconnection links (edges) that form an arbitrary, possibly directed, strongly connected communication topology (digraph). Both distributed algorithms achieve their goals asymptotically and operate iteratively by having each node adapt the (nonnegative) weights on its outgoing edges based on the weights of its incoming links (i.e., based on purely local information). We also provide examples to illustrate the operation, performance, and potential advantages of the proposed algorithms.
\end{abstract}

%
%
%
%
\section{Introduction}\label{sec:introduction}

A distributed system or network consists of a set of subsystems (nodes) that can share information via connection links (edges), forming a directed communication topology (digraph). The successful operation of a distributed system depends on the structure of the digraph which typically proves to be of vital importance for our ability to apply distributed strategies and perform various tasks. Cooperative distributed control algorithms and protocols have received tremendous attention during the last decade by several diverse research communities (e.g., biology, physics, control, communication, and computer science), resulting in many recent advances in consensus-based approaches (see, for example, \cite{2003:jadbabaie_coordination, 2004:Murray, 2005:wei_ren_consensus,2005:Moreau,2006:angeli_stability,2006:XiaoWang,2007:olfati-saber_consensus, 2012:damiano} and references therein). 

In general, the objective of a consensus problem is to have the agents belonging to a group agree upon a certain ({\em a priori} unknown) quantity of interest. When the agents reach agreement, we say that the distributed system reaches consensus. Such tasks include network coordination problems involving self-organization, formation patterns, parallel processing and distributed optimization. One of the most well known consensus problems is the so-called average consensus problem in which agents aim to reach the average of their initial values. The initial value associated with each agent might be, for instance, a sensor measurement of some signal \cite{2004:Sensors}, Bayesian belief of a decision to be taken \cite{2009:Greene}, or the capacity of distributed energy resources for the provisioning of ancillary services \cite{2011:Christoforos-Themis}. Average consensus is closely related to two classes of graphs: weight-balanced (for continuous-time systems) and bistochastic graphs (for discrete-time systems). A weighted graph is balanced if for each node, the sum of the weights of the edges outgoing from that node is equal to the sum of weights of the edges incoming to that node. A bistochastic graph is a weight-balanced graph for which the weights are nonnegative and the sum of the weights for both outgoing and incoming edges (including self weights at each node) is equal to one. In both cases, all edge weights are typically required to be nonzero.

It is shown in \cite{2004:Murray} that average consensus is achieved if the information exchange topology is both strongly connected and balanced, while gossip algorithms \cite{2005:Boyd, 2009:Nedic} and convex optimization \cite{2009:Johansson} admit update matrices which need to be doubly stochastic. These methods have applicability to a variety of topics, such as multi-agent systems, cooperative control, and modeling the behaviour of various phenomena in biology and physics, such as flocking. Since their operation requires weight-balanced and bistochastic digraphs, it is important to be able to distributively transform a weighted digraph to a weight-balanced or bistochastic one, provided that each node is allowed to adjust the weight of its outgoing links accordingly.  

In this paper, we address the problem of designing discrete-time coordination algorithms that allow a networked system to distributively obtain a set of weights that make it weight-balanced or bistochastic. This task is relatively straightforward when the underlying communication topology forms an undirected graph (i.e., when communication links are bi-directional) but more challenging when dealing with a digraph (i.e., when some communication links might be uni-directional). The paper proposes two algorithms that can be used in distributed networks with directed interconnection topologies; the first algorithm leads to a weight-balanced digraph. Even though there exist some earlier approaches in the literature for weight balancing (e.g., \cite{2009:Cortes} presents a finite-time algorithm but does not characterize the number of steps required in the worst case, whereas \cite{2012:Rikos} presents an asymptotic algorithm whose rate is bounded explicitly based on the structure of the graph), this is the first time, to the best of our knowledge, that an asymptotic algorithm of this nature is shown to admit geometric convergence rate.  Work in the literature appears for bistochastic matrix formation as well (e.g., \cite{2011:Christoforos} proposes an asymptotic algorithm with an unspecified rate of convergence). Our second proposed algorithm, a modification of the weight-balancing algorithm, leads to a bistochastic digraph with asymptotic convergence. Under some minor assumptions, this second algorithm can also be shown to admit a geometric rate.

The remainder of the paper is organised as follows. In Section \ref{sec:preliminaries}, we provide necessary notation and background on graph properties. In Section \ref{sec:formulation}, the problem to be solved is formulated, and Sections \ref{sec:resultsA} and \ref{sec:resultsB} present our main results in which we propose two algorithms, one for weight-balancing and one for bistochastic matrix formation. Then, in Section~\ref{sec:examples}, the derived algorithms are demonstrated via illustrative examples and their performance is compared against existing approaches in the literature. Finally, Section \ref{sec:conclusions} presents concluding remarks and future directions. 

%
%
%
%
\section{Notation and Preliminaries}\label{sec:preliminaries}

The sets of real, integer and natural numbers are denoted by $\mathds{R}$, $\mathds{Z}$ and $\mathds{N}$, respectively; their positive orthant is denoted by the subscript $+$ (e.g. $\mathds{R}_{+}$). The symbol $\mathds{N}_0$ denotes the set of non-negative integers. Vectors are denoted by small letters whereas matrices are denoted by capital letters; $A^{-1}$ denotes the inverse of matrix $A$. By $I$ we denote the identity matrix (of appropriate dimensions), whereas by $\mathbb{1}$ we denote a column vector (of appropriate dimension) whose elements are all equal to one. A matrix whose elements are nonnegative, called nonnegative matrix, is denoted by $A \geq 0$ and a matrix whose elements are positive, called positive matrix, is denoted by $A>0$. $\sigma(A)$ denotes the spectrum of matrix $A$, $\lambda(A)$ denotes an element of the spectrum of matrix $A$, and $\rho(A)$ denotes its spectral radius. Notation $\diag(x_{i})$ is used to denote the matrix with elements in the finite set $\{ x_{1}, x_{2}, ..., x_{i}, ... \}$ on the leading diagonal and zeros elsewhere. 

Let the exchange of information between nodes be modeled by a weighted digraph (directed graph) $\mathcal{G}(\mathcal{V}, \mathcal{E},W)$ of order $n$ $(n \geq 2)$, where $\mathcal{V} = \{v_1,v_2,\ldots,v_n\}$ is the set of nodes, $\mathcal{E} \subseteq \mathcal{V} \times \mathcal{V}-\{(v_j, v_j) \; | \; v_j \in \mathcal{V} \}$ is the set of edges, and $W = [w_{ji} ] \in \mathbb{R}_{+}^{n\times n}$ is a weighted ${n\times n}$ adjacency matrix where $w_{ji}$ are nonnegative elements. A directed edge from node $v_i$ to node $v_j$ is denoted by $\varepsilon_{ji} = (v_j, v_i)\in \mathcal{E}$, which represents a directed information exchange link from node $v_i$ to node $v_j$, i.e., it denotes that node $v_j$ can receive information from node $v_i$. Note that the definition of $\mathcal{G}$ excludes self-edges (though self-weights are added when we consider bistochastic digraphs). A directed edge $\varepsilon_{ji} \in \mathcal{E}$ if and only if $w_{ji} > 0$. The graph is undirected if and only if $\varepsilon_{ji} \in \mathcal{E}$ implies $\varepsilon_{ij}  \in \mathcal{E}$. Note that a digraph $\mathcal{G}(\mathcal{V}, \mathcal{E})$ can be thought of as a weighted digraph $\mathcal{G}(\mathcal{V}, \mathcal{E}, W)$ by defining the weighted adjacency matrix $W$ with $w_{ji} =1$ if $\varepsilon_{ji} \in \mathcal{E}$, and $w_{ji} =0$ otherwise.

A digraph is called \emph{strongly} connected if for each pair of vertices $v_j, v_i \in \mathcal{V}$, $v_j \neq v_i$, there exists a directed \emph{path} from $v_j$ to $v_i$, i.e., we can find a sequence of vertices $v_j \equiv v_{l_0}$, $v_{l_1}$, ..., $v_{l_t} \equiv v_i$ such that $(v_{l_{\tau+1}}, v_{l_\tau}) \in \mathcal{E}$ for $\tau = 0, 1, ..., t-1$. All nodes that can transmit information to node $v_j$ directly are said to be in-neighbors of node $v_j$ and belong to the set $\mathcal{N}^{-}_j=\{ v_i\in \mathcal{V} \; | \; \varepsilon_{ji} \in \mathcal{E} \}$. The cardinality of $\mathcal{N}^{-}_j$, is called the \emph{in-degree} of $j$ and it is denoted by $\mathcal{D}^{-}_{j}$. The nodes that receive information from node $j$ comprise its out-neighbors and are denoted by $\mathcal{N}^{+}_j=\{ v_l \in \mathcal{V} \; | \; \varepsilon_{lj} \in \mathcal{E} \}$. The cardinality of $\mathcal{N}^{+}_j$, is called the \emph{out-degree} of $v_j$ and it is denoted by $\mathcal{D}^{+}_{j}$. Given a weighted digraph $\mathcal{G}(\mathcal{V}, \mathcal{E},W)$ of order $n$, the total \emph{in-weight} of node $v_j$ is denoted as $S_j^-$ and is defined by $S_j^- = \sum_{v_i \in \mathcal{N}_j^{-}} w_{ji}$, whereas the total \emph{out-weight} of node $v_j$ is denoted by $S_j^+$ and is defined as $S_j^+ = \sum_{v_l \in \mathcal{N}_j^{+}} w_{lj}$.

\begin{defn}
A weighted digraph $\mathcal{G}(\mathcal{V}, \mathcal{E},W)$ is called weight-balanced if the total in-weight equals the total out-weight for every node $v_j \in V$, i.e., $S_j^- = S_j^+$. A weight-balanced digraph is also called doubly stochastic (bistochastic) if each of its weighted adjacency matrix rows and columns sums to $1$.
\end{defn}

For the discrete-time setup we investigate, we conveniently define the time coordinate so that unity is the time between consecutive iterations. For example $S_j^+[k]$ will denote the value of the out-weight of node $v_j$ at time instant $k$, $k\in \mathds{N}_0$.


%
%
%
%
\section{Problem Formulation}\label{sec:formulation}

Given a digraph $\mathcal{G}(\mathcal{V}, \mathcal{E})$, we want distributed algorithms that allow the nodes to obtain a weight matrix $W = [w_{ji}$ such that the following are achieved.
\begin{enumerate}
\item[(i)] The weighted digraph becomes balanced in a distributed fashion; i.e., a weight matrix $W$ is found such that $w_{ji} > 0$ for each edge $(v_j,v_i) \in E$, $w_{ji}=0$ if $(v_j,v_i) \notin E$, and $S_j^+ = S_j^-$ for every $v_j \in V$.  
\item[(ii)] The weighted digraph becomes bistochastic in a distributed fashion; i.e., a weight matrix $W$ is found with nonnegative diagonal elements $w_{jj} \geq 0$, such that $w_{ji} > 0$ if $(v_j,v_i) \in E$, $w_{ji}=0$ if $(v_j,v_i) \notin E$, $v_j \neq v_i)$, and $w_{jj}+S_j^+ = w_{jj}+S_j^{-} =1$ for every $v_j \in V$. 
\end{enumerate}

%
%
%
%

\section{Distributed algorithm for weight-balancing}\label{sec:resultsA}

\subsection{Description of the algorithm}

Balancing a weighted digraph can be accomplished via a variety of algorithms. We introduce and analyse a distributed cooperative algorithm that exhibits asymptotic convergence and outperforms existing algorithms suggested in the literature \cite{2008:Cortes}. The algorithm achieves weight-balance as long as the underlying digraph is strongly connected (or is a collection of strongly connected digraphs, a necessary and sufficient condition for balancing to be possible \cite{2010:Cortes}). The rate of convergence of the algorithm is geometric and depends exclusively on the structure of the given diagraph and some constant parameters chosen by the nodes.

\emph{Algorithm~1} is an iterative algorithm in which the nodes distributively adjust the weights of their outgoing links such that the digraph becomes asymptotically weight-balanced. We assume that each node observes but cannot set the weights of its incoming links. Given a strongly connected digraph $\mathcal{G}(\mathcal{V}, \mathcal{E})$, the distributed algorithm has each node initialize the weights of all of its outgoing links to unity, i.e., $w_{lj}[0]=1$, $\forall v_l \in \mathcal{N}^{+}_j$. Then, it enters an iterative stage where each node performs the following steps:
\begin{enumerate}
\item It computes its weight imbalance defined by 
\begin{align}\label{eq:imbalance}
x_j[k]\triangleq S_j^-[k]-S_j^+[k].
\end{align}
\item If it is positive (resp. negative), all the weights of its outgoing links are increased (resp. decreased) by an equal amount and proportionally to $x_j[k]$. 
\end{enumerate} 
We discuss why the above distributed algorithm asymptotically leads to weights that balance the graph (and also characterize its rate of convergence) after we describe the algorithm in more detail. For simplicity, we assume that during the execution of the distributed algorithm, the nodes update the weights on their outgoing links in a synchronous\footnote{Even though we do not discuss this issue in the paper, asynchronous operation is not a problem.} fashion. Also note that 2) above implies that $w_{lj}$ for $v_l \in \mathcal{N}_j^+$ will always have the same value.

\begin{algorithm}
\caption{Weight balancing algorithm}
\textbf{Input:} A strongly connected digraph $\mathcal{G}(\mathcal{V}, \mathcal{E})$ with $n=|\mathcal{V}|$ nodes and $m=|\mathcal{E}|$ edges (and no self-loops).\\
\textbf{Initialization:} Each node $v_j\in \mathcal{V}$ \\
\noindent 1) Sets $w_{lj}[0]=1$, $\forall v_l \in \mathcal{N}^{+}_j$. \\
\noindent 2) Chooses $\beta_j\in (0,1)$. \\ 
\textbf{Iteration:} For $k=0,1,2, \ldots$, each node $v_j \in \mathcal{V}$ updates the weights of each of its outgoing links $w_{lj}$, $\forall v_l \in \mathcal{N}^{+}_j$, as
\begin{align}
w_{lj}[k+1]= w_{lj}[k] +\beta_j  \left(\frac{S_j^-[k]}{{\mathcal{D}}^{+}_j} - {w_{lj}[k]}\right). \label{eq:1}
\end{align}
\label{algorithm:1}
\end{algorithm}

The intuition behind the proposed algorithm is that we compare $S_j^-[k]$ with ${S_j^+[k]}={D}^{+}_j{w_{ij}[k]}$. If ${S_j^+[k]}>{S_j^-[k]}$ (resp. ${S_j^+[k]}<{S_j^-[k]}$), then the algorithm reduces (resp. increases) the weights on the outgoing links.

\begin{prop}
\label{prop:1}
If a digraph is strongly connected or is a collection of strongly connected digraphs, \emph{Algorithm~\ref{algorithm:1}} reaches a steady state weight matrix $W^{*}$ that forms a weight-balanced digraph, with geometric convergence rate equal to 
$$
R_{\infty}(P)=-\ln \delta(P), 
$$
where
$$
\delta(P)\triangleq \max\{|\lambda|: \lambda\in \sigma(P)),\lambda \neq 1\}.
$$
\end{prop}

\begin{proof}
First, we observe from equation \eqref{eq:1} that all the outgoing links have the same weight, i.e., $w_{l'j}=w_{lj}$, $\forall v_{l'}, v_l \in \mathcal{N}^{+}_j$ (because they are equal at initialization and they are updated in the same fashion). Hence, from hereafter, we will denote the weight on any outgoing link of node $v_j$ as $w_j$. In order to study the system with update formula \eqref{eq:1} for each node in the graph, we define $w=(w_1 \ \  w_2 \ \ \ldots \ \ w_n)^T$ with $w_j=w_{lj}$ $(v_l\in\mathcal{N}_j^{+})$, and thus we can write the evolution of the weights in matrix form as follows.
\vspace{-0.1cm}
\begin{align}
w[k+1]= P w[k], w[0]=w_0  \label{eq:1matrix}
\end{align}
where $w_0 =\mathbb{1}$ and
\vspace{-0.1cm}
\begin{align}\label{eq:P}
P_{ji}\triangleq 
\begin{cases} 
1-\beta_j , & \text{if $i=j$,} \\
{\beta_{j}}/{{D}^{+}_j} , &\text{if $v_i \in \mathcal{N}^{-}_j$.}
\end{cases}
\end{align}
It should be clear from the above update equation that the weights remain nonnegative during the execution of the algorithm.

Matrix $P$ can be written as $P=I-B+BD^{-1}A$, where $I$ is the identity matrix, $B=\diag(\beta_j)$, $D=\diag(D^{+}_j)$ and $A$ is the adjacency matrix with $a_{ji} =1$ if $\varepsilon_{ji} \in \mathcal{E}$, and $a_{ji} =0$ otherwise. Since $\sigma(D^{-1}A)=\sigma(AD^{-1})$, then $\rho(D^{-1}A)=\rho(AD^{-1})$. In addition, $\rho(AD^{-1})=1$ because matrix $AD^{-1}$ is column stochastic. As a result, $\rho(D^{-1}A)=1$. Also, note that $\bar{P}\triangleq I-B+AD^{-1}B$ is column stochastic and therefore $\rho(\bar{P})=1$. Furthermore, 
\begin{align*}
\rho(\bar{P})&=\rho(\bar{P}B^{-1}DD^{-1}B)=\rho(D^{-1}B\bar{P}B^{-1}D) \\
&=\rho(D^{-1}B(I-B+AD^{-1}B)B^{-1}D) \\
&=\rho(I-B+BD^{-1}A)=\rho(P)=1.
\end{align*}

Since $P$ is a nonnegative matrix, we can ask whether $P$ is primitive,\footnote{A nonnegative matrix is said to be primitive if it is irreducible and has only one eigenvalue of maximum modulus \cite{1985:matrix}.} i.e., whether $P^{m}>0$ for some $m\geq1$. Since the digraph is strongly connected for $0<\beta_j<1$, $\forall v_j \in  \mathcal{V}$ and all the main diagonal entries are positive, we easily conclude that $m\leq n-1$ \cite[\emph{Lemma 8.5.5}]{1985:matrix} and $P$ is primitive. Hence, there is no other eigenvalue with modulus equal to the spectral radius. A sufficient condition for primitivity is that a nonnegative irreducible matrix $A$ has at least one positive diagonal element \cite[Example 8.3.3]{2000:meyer}, which means that some $\beta_j$ can also be set at unity (as long as at least one is strictly smaller than unity). Hence, iteration \eqref{eq:1matrix} has a geometric convergence rate $R_{\infty}(P)=-\ln \delta(P)$, where
$\delta(P)\triangleq \max\{|\lambda|: \lambda\in \sigma(P)),\lambda \neq 1\}$.
In other words, $\delta(P)$ is the second largest of the moduli of the eigenvalues of $P$ (see also \cite{2003:jadbabaie_coordination,2004:Murray, 2005:Moreau}).
\end{proof}

\begin{remark}
If we change the coordinates of $w[k]$ by introducing $z[k]=B^{-1}D w[k]$, then equation \eqref{eq:1matrix} becomes
$z[k+1]=B^{-1}DPD^{-1}Bz[k]=\bar{P}z[k]$,
which has total mass ($\mathbb{1}^T z[k]$) that remains invariant throughout the iteration (and equal to $\mathbb{1}^T z[0]=\mathbb{1}^T B^{-1}D \mathbb{1}$).
\end{remark}

\subsection{Illustrative Example}\label{subsec:example1}

In this illustrative example (borrowed from \cite{2010:Cortes}), we demonstrate the validity of the proposed algorithm in the network depicted in Figure \ref{example_digraph}. In our plots we are typically concerned with the absolute balance defined as
$\varepsilon[k]=\sum_{j=1}^n |x_j[k]|$,
where $x_j[k]$ is given in \eqref{eq:imbalance}. If weight balance is achieved, then $\varepsilon[k]=0$ and $x_j[k]=0$, $\forall v_j \in \mathcal{V}$.

\begin{figure}[h]
\begin{center}
\includegraphics[width=0.34\columnwidth]{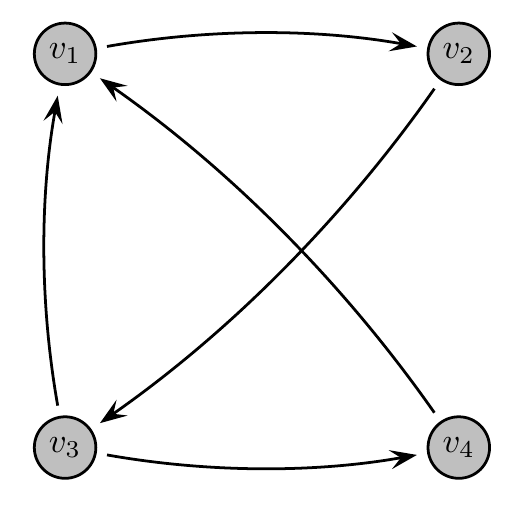}
\caption{A simple digraph which is weight-balanceable, but not bistochasticable due to the absence of self-loops. This example is given in \cite{2010:Cortes} in order to indicate that not all strongly connected digraphs are bistochasticable. }
\label{example_digraph}
\end{center}
\end{figure}

Figure \ref{graph_exampleA} shows the absolute balance of \emph{Algorithm~1} when $\beta_j=0.1, 0.5 \text{ and } 0.9$ for all $v_j\in \mathcal{V}$ as it evolves during the execution of the algorithm. These plots agree with the claims in Proposiion~\ref{prop:1} and validate that the algorithm convergences to a weight-balanced digraph with geometric convergence rate.
\begin{figure}[h!]
\begin{center}
\includegraphics[width=0.79\columnwidth]{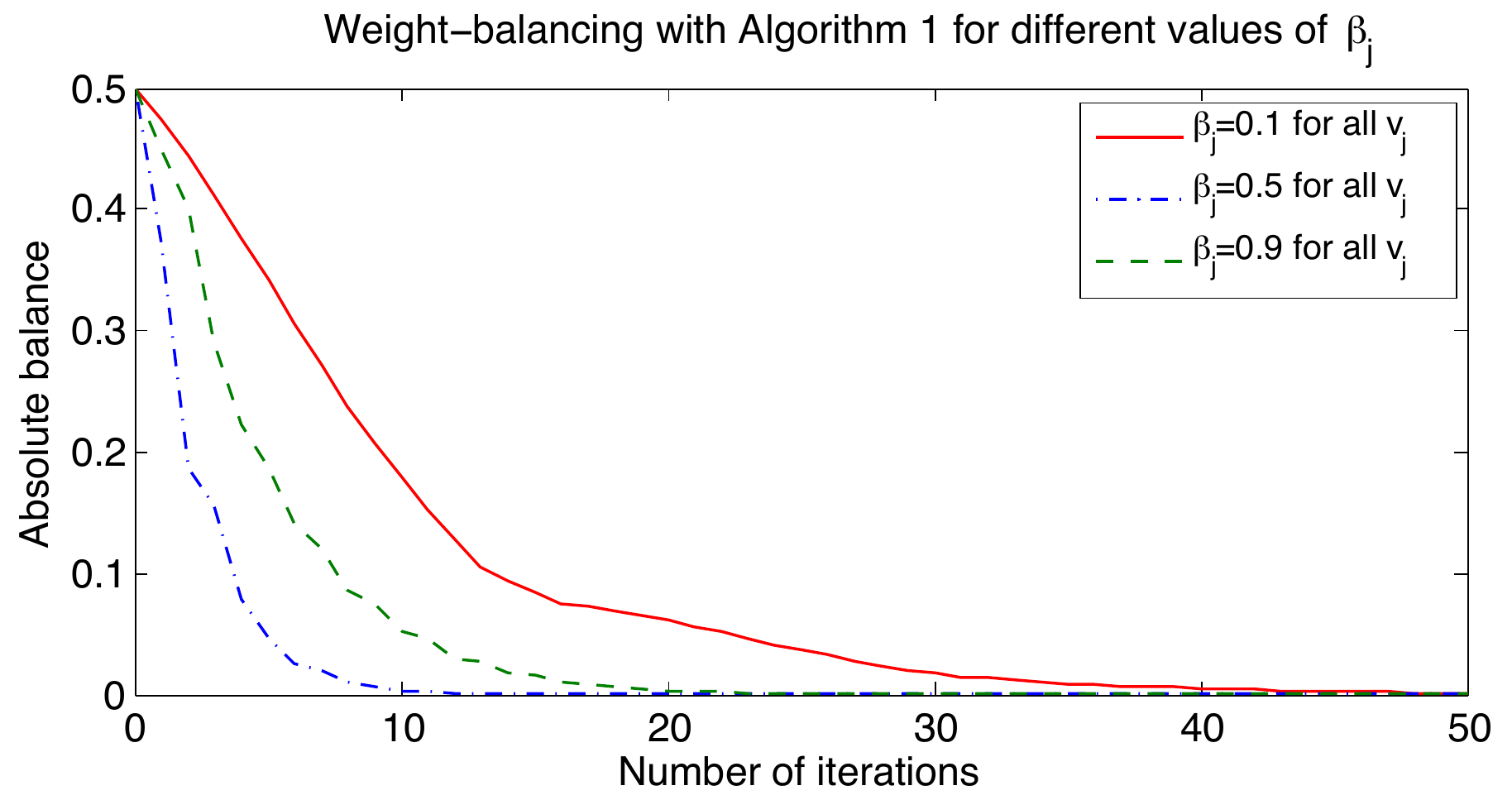}
\caption{Absolute balance for weight-balancing algorithm (\emph{Algorithm~1}) for the digraph depicted in Figure \ref{example_digraph}.}
\label{graph_exampleA}
\end{center}
\end{figure}
In this particular example, with the choice of $\beta_j=0.1, 0.5 \text{ and } 0.9$, the rate discussed in the proof of Proposition~\ref{prop:1} is given by $R_{\infty}\big(P(\beta_j)\big)=0.1204, 0.5180 \text{ and } 0.2524$, respectively. The rates of convergence are also illustrated in Figure~\ref{fig:rates1}, where the convergence rate of Algorithm~1 as obtained in Proposition~\ref{prop:1} and the actual convergence rate of the
algorithm characterized by matrix~\eqref{eq:P}. The theoretical and actual convergence rates coincide.
\begin{figure}[h!]
\begin{center}
\includegraphics[width=0.79\columnwidth]{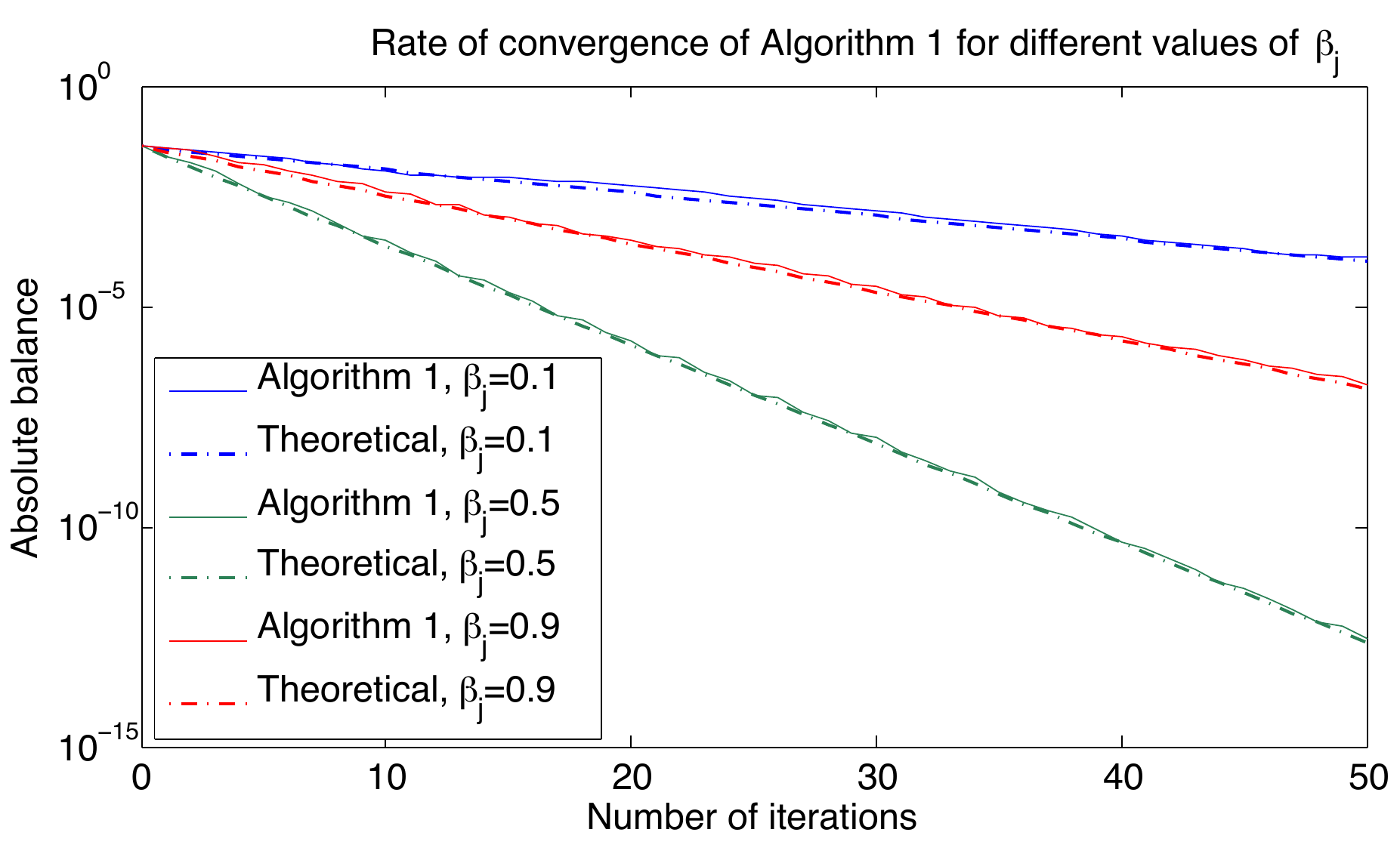}
\caption{Comparison of the theoretical rate of convergence as obtained in Proposition~\ref{prop:1} and the actual convergence rate of the weight-balancing algorithm (\emph{Algorithm~1}) for the digraph depicted in Figure \ref{example_digraph}, for $\beta_j=0.1, 0.5 \text{ and } 0.9$. The dashed lines show the theoretical convergence rate, while the solid lines show the actual convergence rate of the algorithm.}
\label{fig:rates1}
\end{center}
\end{figure}

The final weight-balanced digraph $W^{\star}$ is the same for all three cases mentioned above (because all the $\beta_j$ are identical --- but equal to a different constant each time) and is given by
\begin{align*}
W^{\star}=\begin{bmatrix} 
0 		& 0 		& 0.7143 	& 0.7143 	\\
1.4286 	& 0 		& 0 		& 0		\\
0 		& 1.4286 	& 0 		& 0 		\\
0 		& 0 		&  0.7143 	& 0
\end{bmatrix}.
\end{align*}

\begin{remark}
If $\beta_j=1$, $\forall v_j \in  \mathcal{V}$, then the weighted adjacency matrix $P$ is not necessarily primitive and hence the algorithm does not converge to weights that form a weight-balanced digraph. We can see this in the counterexample depicted in Figure \ref{fig:counter}.
\begin{figure}[h]
\begin{center}
\includegraphics[width=0.34\columnwidth]{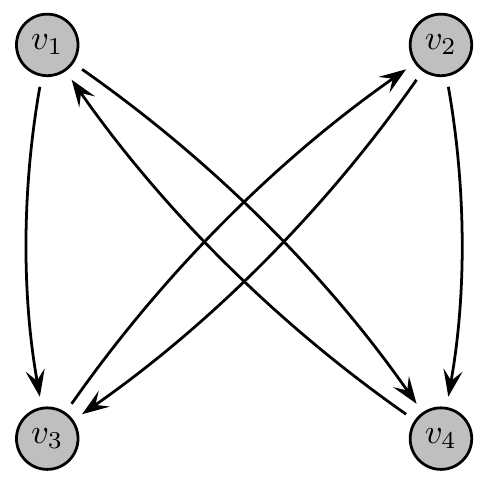}
\caption{A simple digraph which is weight-balanceable, but the adjacency matrix is not primitive.}
\label{fig:counter}
\end{center}
\end{figure}
We run \emph{Algorithm 1} for three different cases: (a) $\beta_j=1$, $\forall v_j \in  \mathcal{V}$, (b) $\beta_j=0.9$, $\forall v_j \in  \mathcal{V}$ and (c) $\beta_1=0.9$ and $\beta_j=1$, $\forall v_j \in  \mathcal{V}-\{v_1\}$. It can be seen from Figure \ref{counterexample} that for the first case the matrix is not primitive and it does not converge, whereas for the other two cases it asymptotically converges (because as long as one of the nodes has $\beta_j<1$, then the update matrix is primitive and the algorithm forms a weight-balanced digraph).
\begin{figure}[h]
\begin{center}
\includegraphics[width=0.79\columnwidth]{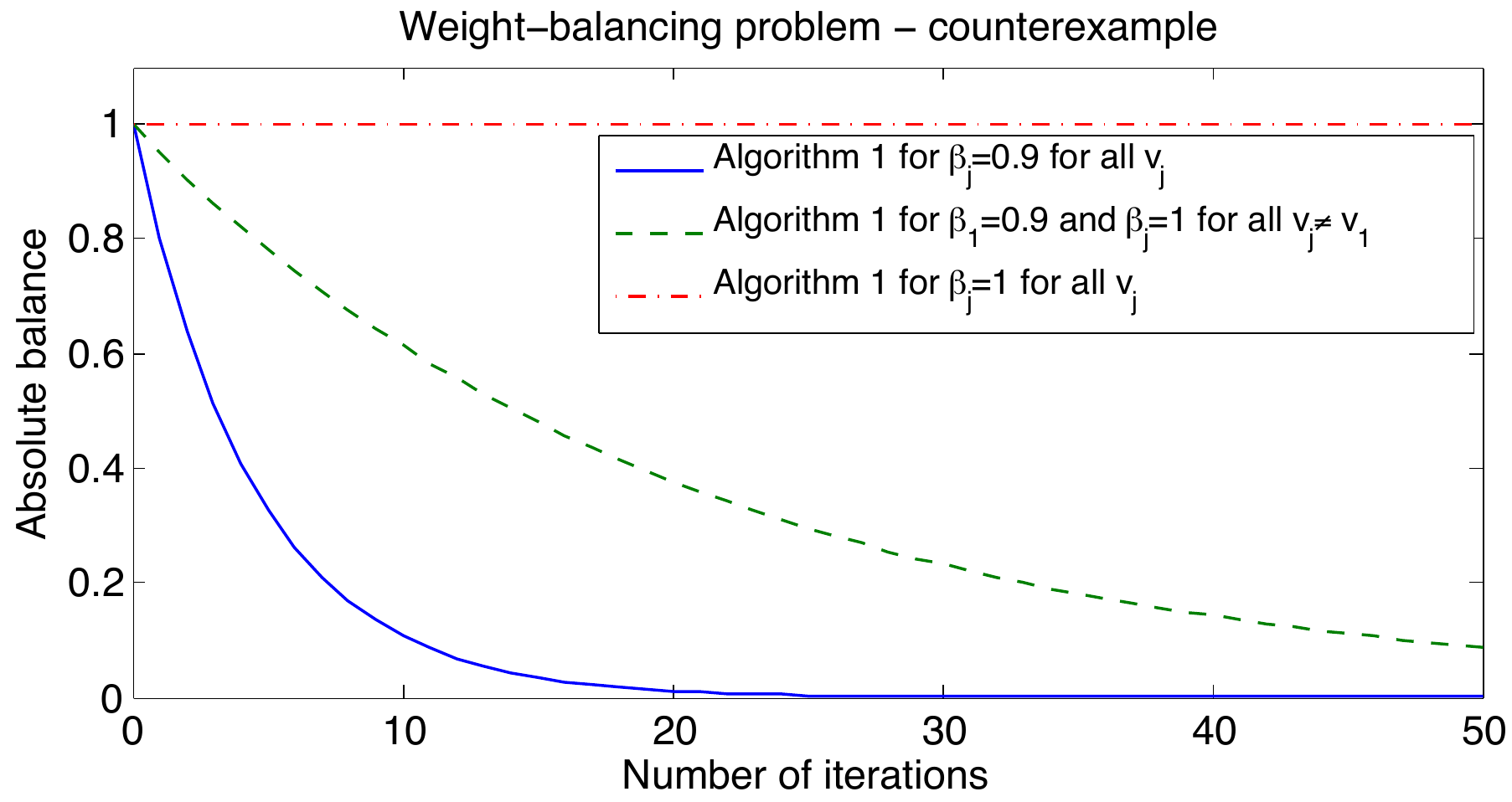}
\caption{A simple digraph which is weight-balanceable,  but when $\beta_j=1$, $\forall v_j \in  \mathcal{V}$, the weighted adjacency matrix is not primitive and hence the algorithm does not converge to weights that form a weight-balanced digraph.}
\label{counterexample}
\end{center}
\end{figure}
%
\end{remark}

\section{Distributed algorithm for bistochastic matrix formation}\label{sec:resultsB}

\subsection{Description of the algorithm}

An algorithm is proposed, herein called \emph{Algorithm~2}, with which a bistochastic adjacency matrix is formed in a distributed fashion.
One extra requirement for \emph{Algorithm~2}, however, is to maintain column stochasticity of the weighted adjacency matrix $W[k]$ for all times $k$, so that it can be used for consensus problems simultaneously from the beginning without the need to form the bistochastic matrix before running an algorithm. More specifically, we obtain a sequence of column stochastic matrices $W[0], W[1], W[2], \ldots, W[k]$ such that $\lim_{k\rightarrow \infty} W[k] = W$ is bistochastic and thus the iteration
$x[k+1]=W[k]x[k+1],~x[0]=x_0$,
reaches average consensus asymptotically \cite{2011:Christoforos}. 

Digraphs that are weight-balanceable do not necessarily admit a doubly stochastic assignment \cite[\emph{Theorem 4.1}]{2010:Cortes}. However, if self-weights are added then 
any strongly connected graph is bistochasticable after adding enough self-loops \cite[\emph{Corollary 4.2}]{2010:Cortes}. 
\emph{Algorithm~2} overcomes this problem with the introduction of nonzero self weights $w_{jj}$ at each node $v_j \in \mathcal{V}$, and their appropriate adjustment in a distributed manner. \emph{Algorithm~2} is described in detail below. 

\begin{algorithm}
\caption{Bistochastic matrix formation algorithm}
\textbf{Input:} A strongly connected digraph $\mathcal{G}(\mathcal{V}, \mathcal{E})$ with $n=|\mathcal{V}|$ nodes and $|\mathcal{E}|$ edges (and no self-loops).\\
\textbf{Initialization:} Set $w_{lj}[0]=1/(1+ \mathcal{D}_j^+)$, $\forall v_l \in \mathcal{N}^{+}_j \cup \{ j \}$. \\
\textbf{Iterations:} For $k=0,1,2, \ldots$, each node $v_j\in V$ updates the weights $w_{lj}$, $v_l \in \mathcal{N}_j^+$, by performing the following steps: \\
\noindent 1) It chooses $\beta_j[k]$ as follows:
\begin{align}\label{eq1:beta}
\beta_j[k] = 
\begin{cases}
\alpha_j\frac{1-S_j^+[k]}{S_j^-[k]-S_j^+[k]}, &  S_j^-[k] > S_j^+[k], \\
\alpha_j,  & \text{otherwise},
\end{cases}
\end{align}
where $\alpha_j\in(0,1)$.\\
\noindent 2) It updates
\begin{align}
w_{lj}[k+1]= w_{lj}[k] +\beta_j[k]  \left(\frac{S_j^-[k]}{{D}^{+}_j} - {w_{lj}[k]}\right), \label{eq1:1}
\end{align}
for all $v_l \in \mathcal{N}_j^{+}$. [This is the same update as \emph{Algorithm~1}, with the difference that the proportionality constant $\beta_j$ can be adapted at each time step $k$, and is chosen so that $S^+_j[k+1] \leq 1$ (so as to ensure that $w_{jj}$ can be chosen in Step 3 to be nonnegative and satisfy $w_{jj} + S^+_j[k+1] =1$).] \\ 
\noindent 3) Then, $w_{jj} \geq 0$ is assigned so that the weighted adjacency matrix retains its column stochasticity; i.e.,
\begin{align}
w_{jj}[k+1] = 1 - \sum_{i \in \mathcal{N}_j^+} w_{ij}[k+1],~\forall v_j \in \mathcal{V}.
\end{align}
\label{algorithm:2}
\end{algorithm}
\begin{prop}\label{prop:2}
If a digraph is strongly connected or is a collection of strongly connected digraphs, then \emph{Algorithm~\ref{algorithm:2}} reaches a steady state weight matrix $W^{*}$ that forms a bistochastic digraph. Furthermore, the weights of all edges in the graph are nonzero.
\end{prop}

\begin{proof}
As before, all the outgoing links have the same weight, i.e., $w_{l'j}=w_{lj}$, $\forall v_{l'}i, v_l \in \mathcal{N}^{+}_j$. Thus, the evolution of the weight $w_j\triangleq w_{lj}$, $\forall v_l \in \mathcal{N}_j^{+}$, can be written in matrix form as follows.
\begin{align}
w[k+1]= P[k] w[k], w[0]=w_0  \label{eq:2matrix}
\end{align}
where 
\begin{align}\label{eq1:P}
P_{ji}[k]\triangleq 
\begin{cases} 
1-\beta_j[k] & \text{if $i=j$,} \\
{\beta_{j}[k]}/{{D}^{+}_j} &\text{if $i \in \mathcal{N}^{-}_j$.}
\end{cases}
\end{align}
In order to make sure that the sum of each column can be made one by choosing a nonnegative self-weight $w_{jj}$, we need to establish (for all $k$) that $S^+_j[k+1] \leq 1$ or $w_{lj}[k+1] \leq 1/D_j^{+}$, for all $v_l \in \mathcal{N}_j^+$. In our updates, there are two cases: (a) $S_j^-[k] \leq S_j^+[k]$, and (b) $S_j^-[k] > S_j^+[k]$; we analyze both cases below. \\
(a) When $S_j^-[k] \leq S_j^+[k]$, then $\beta_j[k]$ can be chosen to be any value $\alpha_j$, $\alpha_j\in(0,1)$. Then, the algorithm is equivalent to Algorithm~1 and admits geometric convergence rates. \\
(b) When $S_j^-[k] > S_j^+[k]$, then $w_{lj}[k+1]$, $\forall v_l \in \mathcal{N}_j^+$, are increased. In order to avoid having  $S^+_j[k+1] > 1$ or $w_{lj}[k+1] > 1/D_j^{+}$ for any $v_l \in \mathcal{N}_j^+$, we choose $\beta_j[k]$ as shown in equation \eqref{eq1:beta}. In this case, for any initial $S_j^{+}[0]$ for which $S_j^-[k] > S_j^+[k]$, from \eqref{eq1:1} we obtain that  
\begin{align*}
S_j^{+}[k+1]&=S_j^{+}[k]+\alpha_j(1-S^{+}_j[k]) \\
&=(1-\alpha_j)S_j^{+}[k]+\alpha_j.
\end{align*}
Hence, it is guaranteed that $S_j^{+}[k+1]<1$, and as a result, $\beta_j[k]\neq 0$ for all $k$. It can be easily shown by perfect induction that after $n$ steps (for which $S_j^-[k+r] > S_j^+[k+r]$ holds for all $r\in\mathbb{Z}_{+}, r\leq n$) we have
\begin{align*}
S_j^{+}[k+n]=1-(1-\alpha_j)^n(1-S^{+}_j[k]),
\end{align*}
that approaches 1 as $n \rightarrow \infty$. Note that someone can choose $\alpha_j\in (0,1)$ closer to $1$ and guarantee that $\beta_j[k]>\epsilon$, $\epsilon>0$, for all the iterations $k$.

\end{proof}

Unlike the case of Algorithm~1, in Algorithm~2 we have a discrete-time switching dynamical system whose stability condition and convergence rate are not as easily characterized. As indicated by the various simulations we have tried (as well as the special case identified by the remark below), it is very likely that the rate of convergence of Algorithm~2 is geometric, however, this has not been formally established thus far.

\begin{prop}\label{prop:3}
If a digraph is strongly connected or is a collection of strongly connected digraphs, then \emph{Algorithm~\ref{algorithm:2}} with initial condition $w_{Lj}[0]=\frac{1}{m(1+ \mathcal{D}_j^+)}$, $\forall v_L \in \mathcal{N}^{+}_j, m\geq |\mathcal{V}|$, reaches a steady state weight matrix $W^{*}$ that forms a bistochastic digraph, with geometric convergence rate equal to $R_{\infty}(P)=-\ln \delta(P)$, where 
\begin{align*}\label{eq:Pa}
P_{ji}[k]\triangleq 
\begin{cases} 
1-\alpha_j & \text{if $i=j$,} \\
{\alpha_{j}}/{{D}^{+}_j} &\text{if $i \in \mathcal{N}^{-}_j$.}
\end{cases}
\end{align*}
Furthermore, the weights of all edges in the graph are nonzero.
\end{prop}

\begin{proof}
If it is possible for each node to know the number of nodes in the graph, or at least an upper bound, then we can set $w_{Lj}[0]=\frac{1}{m(1+ \mathcal{D}_j^+)}$ for all $v_L \in \mathcal{N}^{+}_j$, where $m\geq |\mathcal{V}|$. Hence, we establish that $S_j^{-}[0]<1$ and $S_j^{+}[0]<1$ for all $v_j\in \mathcal{V}$.
By equation \eqref{eq1:1} we have that 
$$
S_j^{+}[k+1]=(1-\beta_j[k])S_j^{+}[k]+\beta_j[k]S_j^{-}[k]
$$ 
which guarantees that $S_j^{+}[k]<1$ and that $\beta_j[k]=\alpha_j$, for all $k$. The self-weights are chosen so that we make each column sum up to one (and consequently the row sums, once the algorithm converges). Therefore, with these initial conditions we can choose $\beta_j[k]$ to be constant (and equal to $\alpha_j\in(0,1)$) throughout the execution of the algorithm, and hence, admit geometric convergence rate (because we are essentially back to Algorithm~1 with the only difference that we also calculate the value of the self-weight (equal to $1-\mathcal{D}_j^+ w_j[k]$) each time). Therefore, Algorithm~2 has a geometric convergence rate $R_{\infty}(P)=-\ln \delta(P)$, as in Algorithm~1.
\end{proof}

\begin{remark}
As before, if we change the coordinates of $w[k]$ by introducing $z[k]=B^{-1}D w[k]$, then equation \eqref{eq:2matrix} becomes
$z[k+1]=B^{-1}DP[k]BD^{-1}z[k]=\bar{P}[k]z[k]$,
which has total mass ($\mathbb{1}^T z[k]$) that remains invariant.
\end{remark}

\subsection{Illustrative Example}

We consider a random strongly connected graph consisting of $50$ nodes. The quantity of interest in this case is the absolute balance of the graph at time step $k$; the absolute balance $Ab[k]$ is defined as
$$
Ab[k] = \sum_{v_j \in \mathcal{V}} \left | 1 - \sum_{v_i \in \mathcal{N}_j^-} w_{ji}[k] \right | \; .
$$
Since the weight matrix at each time step is column stochastic by construction, $Ab[k]$ effectively measures the distance of the weight matrix at time step $k$ from a bistochastic matrix. As it can be seen from Figure \ref{graph_exampleC1}, the algorithm asymptotically converges to a bistochastic adjacency matrix for different values of $\alpha_j$ (in the simulations $\alpha_j$ is chosen to be $0.9, 0.5 \text{ and } 0.1$).

\begin{figure}[h]
\begin{center}
\includegraphics[width=0.79\columnwidth]{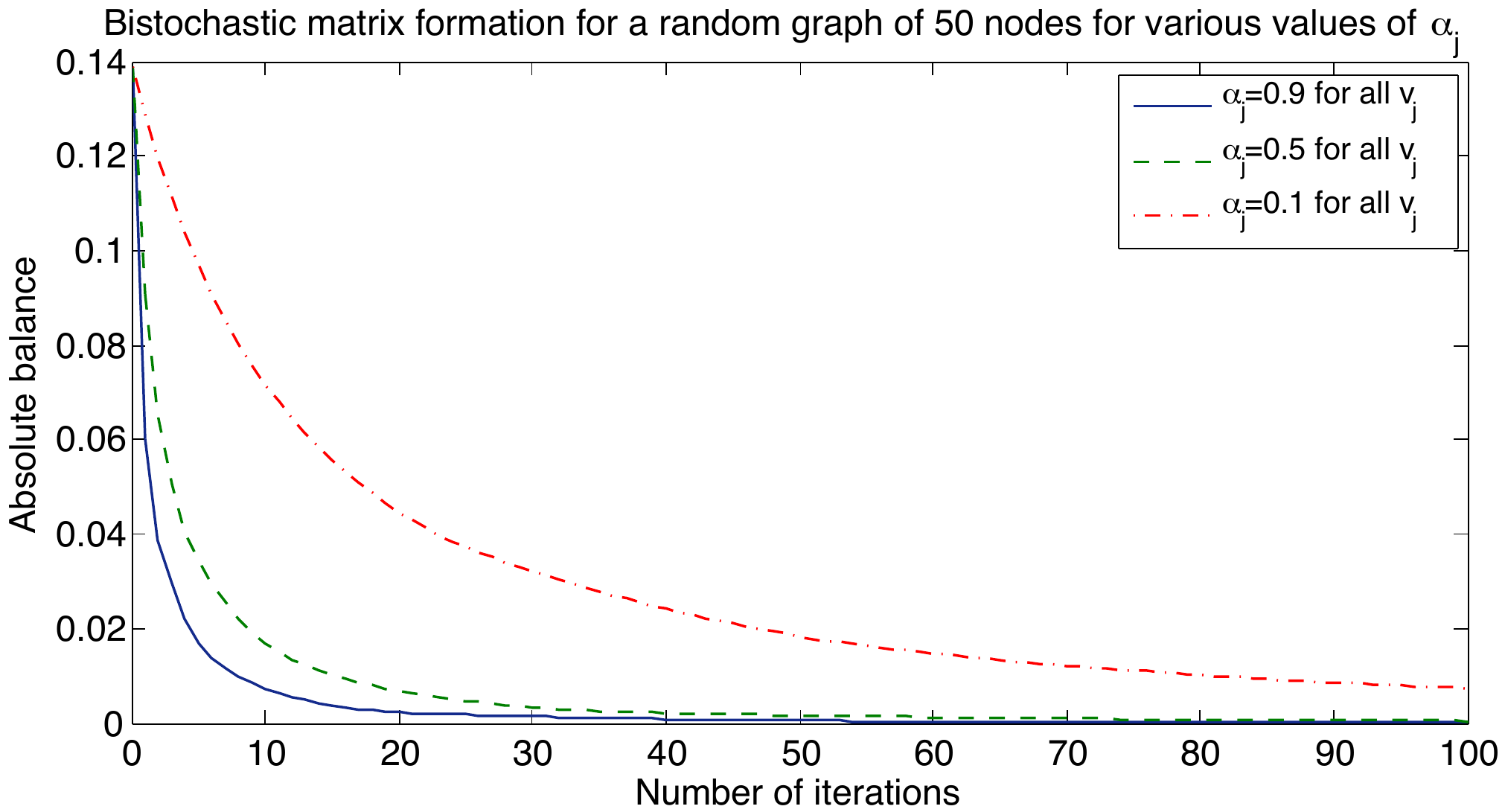}
\caption{Absolute balance for bistochastic formation algorithm (\emph{Algorithm~2}) for different values of $\alpha_j$ $\forall v_j\in \mathcal{V}$. It can be observed that \emph{Algorithm~2} converges to a bistochastic adjacency matrix asymptotically.}
\label{graph_exampleC1}
\end{center}
\end{figure}

%
%
%
%
\section{Comparisons}\label{sec:examples}


\subsection{Weight-Balanced Matrix Formation}

Here we run the algorithm for larger graphs (of size $n=50$) and we compare the performance of our algorithm against the \emph{imbalance-correcting algorithm} in \cite{2009:Cortes} in which every node $v_j$ adds all of its weight imbalance $x_j$ to the outgoing node with the lowest weight $w$.

\begin{figure}[h]
\begin{center}
\includegraphics[width=0.79\columnwidth]{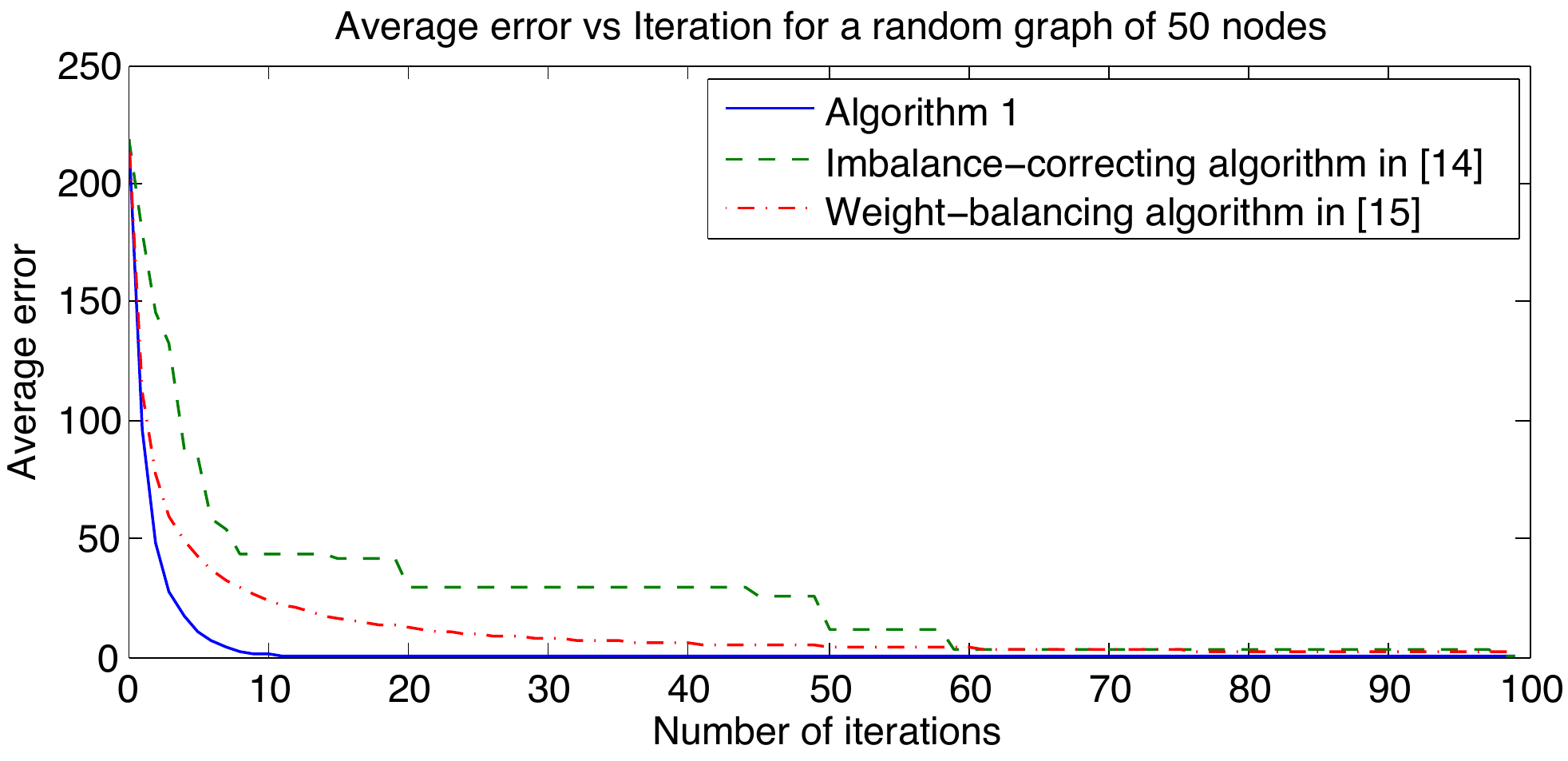}
\includegraphics[width=0.79\columnwidth]{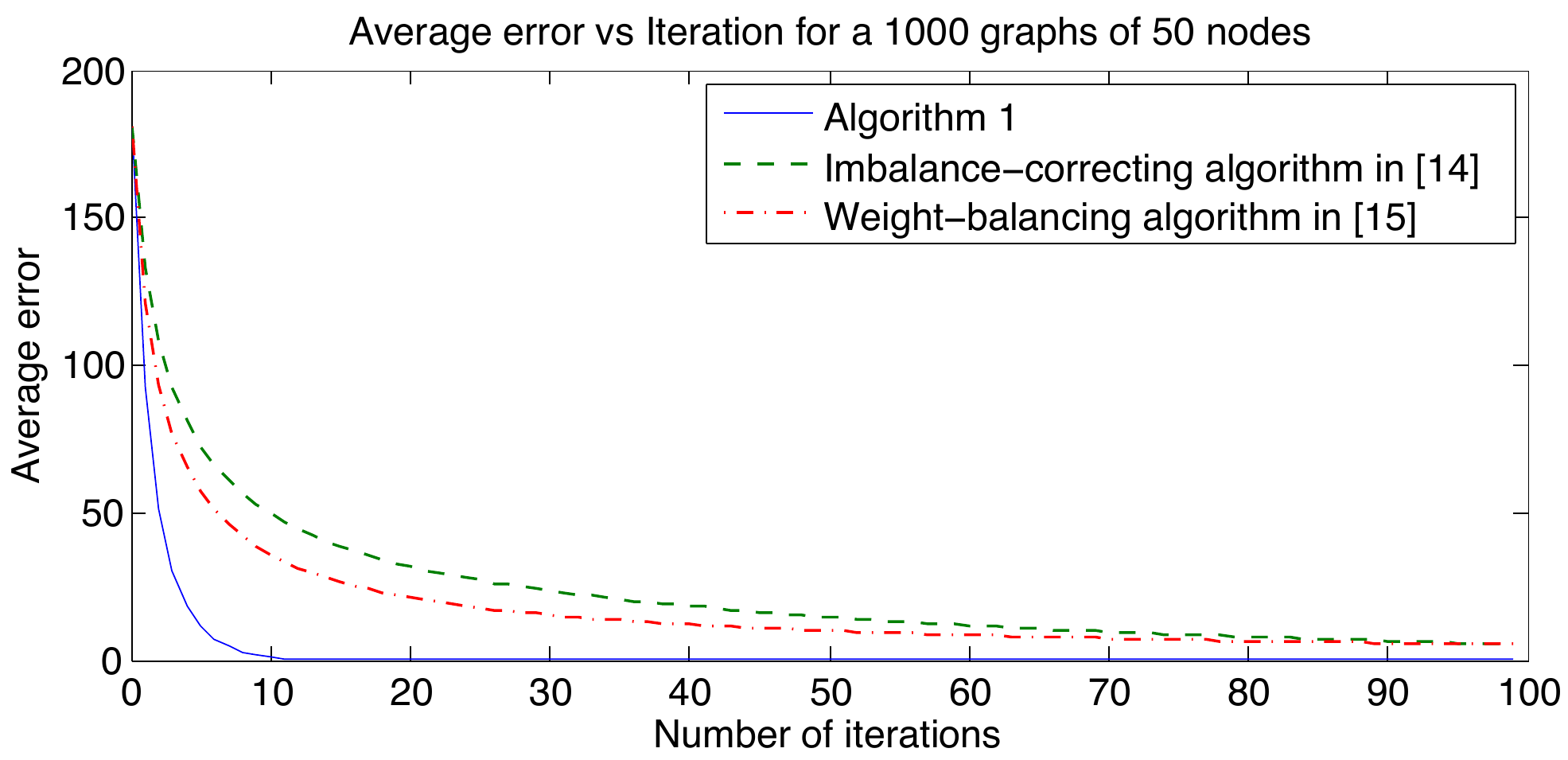}
\caption{
\emph{Top figure:} Absolute balance plotted against the number of iterations for a random graph of 50 nodes for the distributed weight-balancing (\emph{Algorithm~1}) algorithm, the imbalance-correcting algorithm \cite{2009:Cortes}, and the weight-balancing algorithm proposed in \cite{2012:Rikos}.
\emph{Bottom figure:} Average balance plotted against the number of iterations for a 1000 random graphs of 50 nodes for the distributed weight-balancing (\emph{Algorithm~1}), the imbalance-correcting algorithm \cite{2009:Cortes}, and the weight-balancing algorithm proposed in \cite{2012:Rikos}.
}
\label{graph_comparisonA1}
\end{center}
\end{figure}

The suggested weight-balancing algorithm (\emph{Algorithm~1}) shows geometric convergence and outperforms the algorithm suggested in \cite{2009:Cortes}, as shown in Figure \ref{graph_comparisonA1}.

\subsection{Bistochastic Matrix Formation}

Here we run the algorithm for larger graphs (of size $n=50$) and we compare the performance of our algorithm against a distributed algorithm suggested in \cite{2011:Christoforos} in which every node $v_j\in \mathcal{V}$ first chooses a particular weight $w_{jj}[k] \in (0,1)$ for its self weight and it sets the weights of its outgoing links to be $w_{lj}[k]=c_{lj}(1-w_{jj}[k])$, where $c_{lj}>0$, $\forall v_l \in \mathcal{N}_j^{+}$, and $\sum_{v_l \in \mathcal{N}_j^+}c_{lj}=1$.

\begin{figure}[h]
\begin{center}
\includegraphics[width=0.79\columnwidth]{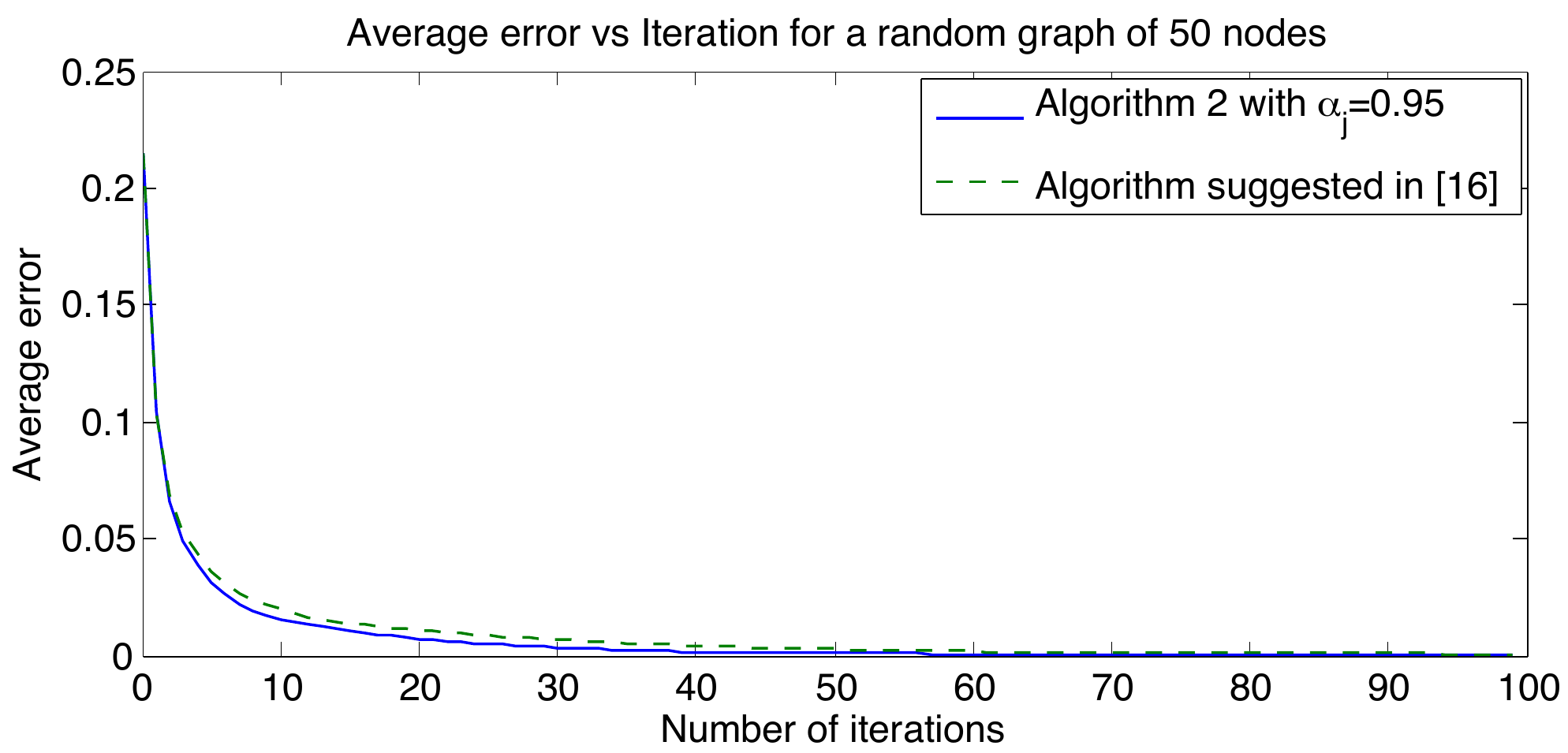}
\includegraphics[width=0.79\columnwidth]{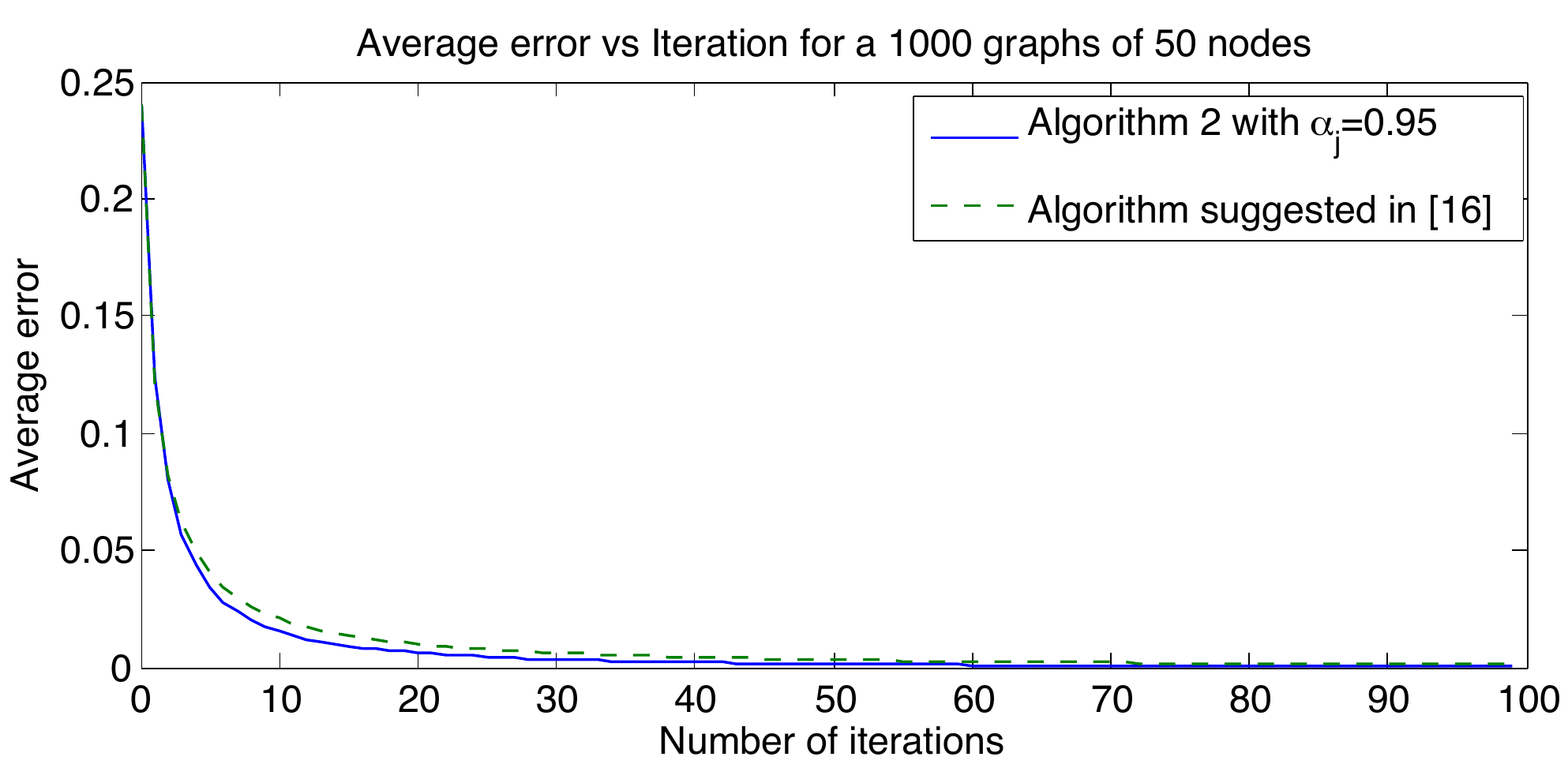}
\caption{
\emph{Top figure:} Absolute balance for bistochastic formation algorithms (\emph{Algorithm~2}) and a distributed algorithm suggested in \cite{2011:Christoforos} for a random strongly connected graph consisting of 50 nodes.
\emph{Bottom figure:} Average balance plotted against the number of iterations for a 100 random strongly connected graphs of 50 nodes for the distributed bistochastic formation algorithms (\emph{Algorithm~2}) and a distributed algorithm suggested in \cite{2011:Christoforos}.
}
\label{graph_comparisonB1}
\end{center}
\end{figure}

The suggested bistochastic formation algorithm (\emph{Algorithm~2}) shows asymptotic convergence and slightly outperforms the algorithm suggested in \cite{2011:Christoforos} for $\alpha_j=0.95$, $\forall v_j \in \mathcal{V}$, as shown in Figure \ref{graph_comparisonB1}.

%
%
%
%
\section{Conclusions and Future Directions}\label{sec:conclusions}

In this paper we have developed two iterative algorithms: one for balancing a weighted digraph and one for forming a bistochastic adjacency matrix in a digraph. Both algorithms are distributed and the second one is a direct extension of the first one. The weight-balancing algorithm is  asymptotic  and is shown to admit geometric convergence rate.  The second algorithm, a modification of the weight-balancing algorithm, leads to a bistochastic digraph with asymptotic convergence admitting geometric convergence rate for a certain set of initial values. The two algorithms perform very well compared to existing algorithms, as illustrated in the examples.

Future work includes the analysis of the behavior of the suggested algorithms in the presence of delays and changing topology.

\bibliographystyle{IEEEtran}
\bibliography{bibliografia_consensus}

\end{document}